\documentclass[sigconf]{acmart}
\copyrightyear{2026}
\acmYear{2026}
\setcopyright{cc}
\setcctype{by-nc-nd}
\acmConference[ICMI '26]{INTERNATIONAL CONFERENCE ON MULTIMODAL INTERACTION}{October 05--09, 2026}{Napoli, Italy}
\acmBooktitle{INTERNATIONAL CONFERENCE ON MULTIMODAL INTERACTION (ICMI '26), October 05--09, 2026, Napoli, Italy}
\acmDOI{10.1145/3776574.3832754}
\acmISBN{979-8-4007-2318-6/2026/10}

\usepackage{
  bm,
  mathtools,
}
\usepackage{colortbl}

\usepackage{cleveref}

\definecolor{cmd}{HTML}{E67E22}
\def\vb#1{\if#1\relax\bm{#1}\else\mathbf{#1}\fi}

\begin{document}

\title{MoDAl: Self-Supervised Neural Modality Discovery via Decorrelation for Speech Neuroprosthesis}

\author{Yuanhao Chen}
\email{yc.th@dartmouth.edu}
\orcid{0009-0009-9689-0514}
\affiliation{%
  \institution{Dartmouth College}
  \city{Hanover}
  \state{NH}
  \country{USA}
}
\author{Peter Chin}
\orcid{0000-0002-1913-4223}
\email{pc@dartmouth.edu}
\affiliation{%
  \institution{Dartmouth College}
  \city{Hanover}
  \state{NH}
  \country{USA}
}

% \renewcommand{\shortauthors}{Chen et al.}

\begin{abstract}
  Speech neuroprosthesis systems decode intended speech from neural activity in the absence of audible output, offering a path to restoring communication for individuals with speech-impairing conditions.
  Current approaches decode predominantly from motor cortical areas, discarding others---such as area~44, part of Broca's area---that may encode complementary linguistic information.
  We introduce \textbf{MoDAl} (\underline{Mo}dality \underline{D}ecorrelation and \underline{Al}ignment), a framework that discovers complementary neural modalities through the interplay of two objectives in a shared projection space.
  A contrastive loss aligns each of several parallel brain encoders with the text embeddings of a pretrained large language model (LLM), while a decorrelation loss prevents the encoders from coalescing to duplicative representations.
  We prove that these objectives are in productive tension: Contrastive alignment induces transitive modality coalescence, which decorrelation must counteract for the framework to discover diverse neurolinguistic modalities.
  On the Brain-to-Text Benchmark~'24, MoDAl reduces word error rate (WER) from 26.3\% to 21.6\% compared to the previous best end-to-end method, with the gain from incorporating previously discarded area~44 signals arising entirely from the decorrelation mechanism.
  Analysis of the discovered modalities reveals functional specialization: Encoders receiving area~44 input capture structural and syntactic properties (sentence length, grammatical voice, wh-words), consistent with the neurolinguistic understanding of Broca's area.
\end{abstract}





\maketitle

\section{Introduction}

Speech is an essential modality of human communication, and restoring it for individuals who have lost the ability to speak is a central goal of assistive technology.
Speech neuroprosthesis---the decoding of intended speech from neural activity in the absence of audible output---has emerged as a promising direction for brain-computer interfaces (BCIs) that could restore communication to people with conditions such as amyotrophic lateral sclerosis (ALS) or locked-in syndrome~\cite{willettHighperformanceSpeechNeuroprosthesis2023,metzgerHighperformanceNeuroprosthesisSpeech2023,silvaSpeechNeuroprosthesis2024}.
As an interface that transduces neural signals into language, speech neuroprosthesis is inherently a multimodal interaction problem and directly addresses the inclusivity goals of the human-computer interaction (HCI) community: enabling natural, high-bandwidth communication for users who cannot access conventional speech or gesture-based interfaces.

Current speech neuroprosthesis systems decode speech almost exclusively from motor signals, leveraging the well-characterized mapping between neural activity and the movements of the tongue, lips, jaw, and larynx~\cite{willettHighperformanceSpeechNeuroprosthesis2023,cardAccurateRapidlyCalibrating2024,metzgerHighperformanceNeuroprosthesisSpeech2023}.
However, speech production involves more than articulators.
Broca's area (including area~44) has long been associated with higher-level language functions including syntactic processing, phonological planning, and semantic retrieval~\cite{flinkerRedefiningRoleBrocas2015,verwoertWholebrainDynamicsArticulatory2025}, yet its signals are discarded by all prior work on the Brain-to-Text Benchmark~'24 because they do not improve phoneme decoding~\cite{willettHighperformanceSpeechNeuroprosthesis2023,fengEndtoEndFrameworkInvasive2024,liBraintotextDecodingContextaware2025,zhangDecodingInnerSpeech2025}.
This raises a question: \emph{Can area~44 signals contribute to speech neuroprosthesis if the system is designed to discover and exploit \textbf{complementary} neurolinguistic modalities rather than treating all neural signals as a single articulatory stream?}

In standard multimodal contrastive learning~\cite{radfordLearningTransferableVisual2021,girdhar2023imagebind}, the modalities---text, image, audio, video---are predefined, and each has its own dedicated encoder.
Decorrelation between modalities is not a concern because the inputs are already categorically distinct.
In the neural setting, by contrast, the cortical signals from area~6v and area~44 are recorded from physically adjacent, intercorrelated electrode arrays, and there is no a priori labeling of which neural patterns constitute distinct ``modalities.''
\Citet{willettHighperformanceSpeechNeuroprosthesis2023} show that area~44 electrodes do encode sparse articulatory information, so simply feeding all channels into a single encoder can lead to redundancy with area~6v signals rather than complementary specialization.

We introduce \textbf{MoDAl} (\underline{Mo}dality \underline{D}ecorrelation and \underline{Al}ignment), a framework that addresses this challenge through the interplay of two objectives in a shared projection space.
A contrastive loss aligns each parallel brain encoder with the text embeddings of a pretrained large language model (LLM), anchoring all modalities to a common linguistic reference.
We prove that this alignment creates a transitive pressure for all modality projections to coalesce (Proposition~\ref{prop:alignment}), which would collapse the encoders into redundant copies.
A decorrelation loss counteracts this coalescence by penalizing feature-wise correlation between encoder pairs (Proposition~\ref{prop:decorr}).
By operating in a shared projection space, the two losses are kept in productive tension: Contrastive alignment prevents the decorrelation projector from collapsing trivially, while decorrelation prevents the encoders from coalescing.

Applying MoDAl to the Brain-to-Text Benchmark~'24 dataset, we show that three parallel brain encoders---trained end-to-end with an LLM decoder---discover complementary modalities that reduce word error rate (WER) from 26.3\% to 21.6\%, improving over the previous best end-to-end (E2E) method by 4.7 percentage points on the test set and approaching the performance of cascaded systems that require separately trained language models.
Crucially, the improvement from adding the previously discarded area~44 signals (MoDAl-1 $\to$ MoDAl-Full) is 0.8 percentage points---a statistically significant gain that arises entirely from the decorrelation mechanism, since without it, area~44 adds no benefit.
Linear probes reveal that the novel encoders specialize in structural and syntactic properties (sentence length, voice, wh-words) rather than articulatory features, consistent with the neurolinguistic understanding of Broca's area.
Our contributions are as follows:
\begin{enumerate}
  \item We propose MoDAl, a framework for self-supervised neural modality discovery that uses contrastive alignment and decorrelation in a shared projection space. We prove that these objectives are in productive tension: Contrastive alignment induces transitive modality coalescence, which decorrelation must counteract to encourage complementary specialization.
  \item We demonstrate that area~44 (part of Broca's area) signals, discarded by all prior work on this dataset, carry complementary linguistic information that MoDAl can discover and exploit, reducing WER by 4.7 percentage points relative to the previous best E2E method.
  \item We provide a detailed analysis of the discovered modalities through cross-correlation measurement, retrieval, canonical correlation analysis, and linear probing for linguistic properties, showing that the parallel encoders specialize in different linguistic features consistent with the known functional roles of their respective brain areas.
\end{enumerate}

\section{Related Work}

\subsection{Speech Neuroprosthesis}

Intracortical BCIs for speech decoding have advanced rapidly in recent years.
\Citet{willettHighperformanceSpeechNeuroprosthesis2023} demonstrated the first high-performance speech neuroprosthesis for a participant with ALS, using an recurrent neural network (RNN) decoder on area~6v neural activity followed by an $n$-gram language model, achieving 23.8\% WER on a 125,000-word vocabulary in real time.
\Citet{metzgerHighperformanceNeuroprosthesisSpeech2023} developed a similar cascaded system for a participant with anarthria.
\Citet{cardAccurateRapidlyCalibrating2024} introduced rapid calibration methods that reduce the data requirements for new users.
These cascaded systems separate phoneme decoding from language modeling, requiring independently trained components and beam search re-scoring at inference time.

End-to-end approaches that directly map neural signals to text through a single differentiable model have emerged as an alternative.
\Citet{fengEndtoEndFrameworkInvasive2024} proposed the first E2E framework for this dataset, using a gated recurrent unit (GRU) encoder with a LLM decoder fine-tuned via quantized low-rank adaptation (QLoRA).
% , achieving 26.3\% test WER.
% \Citet{liBraintotextDecodingContextaware2025} improved the cascaded approach by incorporating context-aware neural representations with diphone and triphone decoding.
\Citet{zhangDecodingInnerSpeech2025} introduced BIT, using a transformer encoder pretrained on cross-species motor signals to achieve strong results in both cascaded and E2E settings.
\Citet{fiedlerTeachingWav2Vec2Language2025} adapted the self-supervised wav2vec~2.0 speech framework~\cite{baevskiWav2vec20Framework2020a} to neural signals, treating brain activity as a ``language'' to be learned.

All of these methods decode exclusively from ventral premotor cortex (area~6v), which encodes articulator movements.
None exploit signals from area~44 (part of Broca's area), which \citet{willettHighperformanceSpeechNeuroprosthesis2023} found to degrade phoneme decoding.
However, neuroscience evidence suggests that Broca's area plays a role in higher-level language processing---including syntactic structure, semantic retrieval, and phonological planning---distinct from articulatory representations in the motor cortex~\cite{flinkerRedefiningRoleBrocas2015,verwoertWholebrainDynamicsArticulatory2025}.
\Citet{verwoertWholebrainDynamicsArticulatory2025} recently showed through whole-brain intracranial recordings that the inferior frontal cortex (containing Broca's area) was correlated almost exclusively with semantic rather than articulatory or acoustic speech representations. They noted that its indirect connection to low-level speech features \emph{``limit[s] its use for an articulatory-based speech neuroprosthesis.''}
Our work tests whether a framework designed for modality \emph{discovery}---rather than articulatory decoding---can recover this complementary information.

% For a broader perspective on the speech neuroprosthesis landscape, including non-intracortical approaches such as electrocorticography (ECoG) and non-invasive recordings, we refer readers to \citet{silvaSpeechNeuroprosthesis2024} and \citet{tangSemanticReconstructionContinuous2023}.

\subsection{Multimodal Learning}

Contrastive learning has become the dominant approach for aligning representations across modalities.
CLIP~\cite{radfordLearningTransferableVisual2021} demonstrated that image and text encoders trained with a symmetric InfoNCE loss
% on web-scale paired data
produce a joint embedding space with strong zero-shot transfer capabilities.
ImageBind~\cite{girdhar2023imagebind} explores using image as an anchor modality, showing that aligning each new modality (audio, depth, thermal, inertial measurement unit) to image via pairwise contrastive losses induces \emph{emergent} alignment between modalities that were never paired during training---a transitive property of shared embedding spaces.
% ViT-Lens~\cite{leiViTLensOmnimodalRepresentations2024} further scaled this idea by repurposing pretrained vision transformers as encoders for novel modalities including 3D point clouds, tactile, and electroencephalography (EEG) data.
% VAST~\cite{chenVASTVisionAudioSubtitleText2023} unifies video, audio, subtitle, and text in a single representation.

A key assumption in these methods is that the modalities are \emph{predefined}: Each has a distinct physical origin (camera, microphone, depth sensor), dedicated encoder, and typically abundant paired data.
Decorrelation between modalities is not needed because the inputs are already categorically distinct---an image and its caption share semantic content but are never at risk of encoding identical low-level features.
In our setting, the parallel brain encoders receive signals from physically adjacent, intercorrelated cortical arrays~\cite{willettHighperformanceSpeechNeuroprosthesis2023}, making inter-modality redundancy the default outcome rather than an exceptional failure mode.

Self-supervised decorrelation methods provide the complementary tool we need.
Barlow Twins~\cite{zbontarBarlowTwinsSelfSupervised2021} introduced a cross-correlation-based loss that drives augmented views of the same input toward invariant but non-redundant representations by pushing the cross-correlation matrix toward the identity.
Our decorrelation loss adapts this idea to the inter-modality setting: Instead of targeting identity (invariance), we target zero diagonal (decorrelation), and we compute the cross-correlation \emph{between} different modality encoders rather than between augmented views.
Moreover, we show that this loss must operate in a shared projection space with the contrastive loss to be effective; a separate decorrelation projector collapses trivially (Table~\ref{tab:crosscorr}).

MoDAl combines these two traditions---contrastive alignment from multimodal learning and decorrelation from self-supervised learning---to address a problem that neither solves alone: discovering complementary modalities from intercorrelated neural signals without predefined modality labels.

\section{Method}

\subsection{Dataset}

We use the Brain-to-Text Benchmark~'24 dataset~\cite{willettHighperformanceSpeechNeuroprosthesis2023}, which contains attempted speech recordings from participant T12, a woman with bulbar-onset ALS and unable to produce intelligible speech.
The dataset comprises 12,100 large-vocabulary sentences collected across 25 recording sessions spanning approximately four months.

\paragraph{Preprocessing}
Neural activity was recorded from four 64-electrode Utah microelectrode arrays (256 electrodes total): two implanted in ventral premotor cortex (area~6v) and two in the ventral portion of area~44, part of Broca's area.
For each cortical area, two signal types are provided in 20\,ms time bins: threshold crossing counts $\vb c^{(\mathrm{44})}, \vb c^{(\mathrm{6v})}$ and spike band power $\vb p^{(\mathrm{44})}, \vb p^{(\mathrm{6v})}$.
Blockwise z-scoring is applied to mitigate recording nonstationarities.
Following \citet{willettHighperformanceSpeechNeuroprosthesis2023}, we stack each sequence with a kernel of $k{=}14$ bins and stride $s{=}4$ to yield patches of 280\,ms with 200\,ms overlap, giving $D_{\mathrm{in}}=64k=896$ features per signal type.

\paragraph{Previously Discarded Area~44 Signals}
Previous studies on this dataset decode speech exclusively from area~6v.
\Citet{willettHighperformanceSpeechNeuroprosthesis2023} report that area~44 contains little production-related information and that including it degrades decoding performance; subsequent works have likewise discarded area~44 entirely~\cite{fengEndtoEndFrameworkInvasive2024,liBraintotextDecodingContextaware2025,zhangDecodingInnerSpeech2025,fiedlerTeachingWav2Vec2Language2025}.
A central question of this work is whether area~44 signals carry complementary linguistic information that can be recovered through self-supervised modality discovery and leveraged by an LLM decoder.

\subsection{Architecture}

\begin{figure}
  \includegraphics[width=\linewidth]{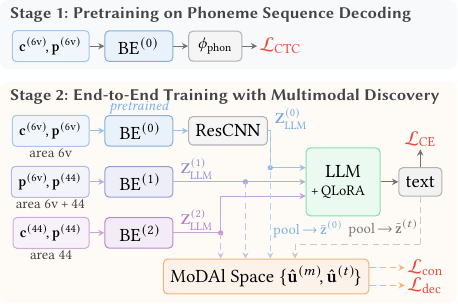}
  \caption{MoDAl architecture. Stage~1 pretrains brain encoder $\operatorname{BE}^{(0)}$ on phonemes with CTC loss. In Stage~2, three parallel brain encoders feed a QLoRA-adapted LLM decoder (solid), while branching into the shared MoDAl projection space with text (dashed) for contrastive alignment and decorrelation.}
  \Description[Architecture diagram of the MoDAl framework showing two training stages.]{Two gray boxes labeled Stage 1 and Stage 2. Stage 1 is a linear pipeline from area 6v input through brain encoder BE0 to CTC loss. Stage 2 shows three color-coded parallel brain encoders (blue, blue-purple, purple) receiving different cortical inputs, converging via solid colored arrows into an LLM decoder connected bidirectionally to a text node. At junction dots on each path, dashed vertical lines branch into an orange MoDAl Space box below. A dashed line from the text node also enters the MoDAl space. Two dashed red arrows exit the MoDAl box rightward to contrastive and decorrelation losses.}
  \label{fig:architecture}
\end{figure}

\subsubsection{Articulatory Brain Encoder}
The only component trained in Stage~1 is a brain encoder that maps area~6v neural activity to articulatory features for phoneme decoding.
GRU-based encoders were adopted following prior work on this benchmark, which found that GRUs outperform LSTM and transformer alternatives~\cite{liBraintotextDecodingContextaware2025,fiedlerTeachingWav2Vec2Language2025}.
We additionally apply channel-wise locked dropout~\cite{galDropoutBayesianApproximation2016} to the input features and layer normalization~\cite{baLayerNormalization2016} to the recurrent outputs, which we find improves generalization to held-out sessions.

To accommodate neural nonstationarities across recording sessions, we apply a day-specific affine input projection indexed by session day $d$ followed by Softsign activation:
\[
  \vb X' = \phi_{\mathrm{in}}^{(d)}(\vb X) = \operatorname{Softsign}(\vb X \vb W_{\mathrm{in}}^{(d)\top} + \vb b_{\mathrm{in}}^{(d)}) \in \mathbb{R}^{B\times T\times 2D_{\mathrm{in}}},
\]
where $\vb X$ concatenates $\vb c^{(\mathrm{6v})}$ and $\vb p^{(\mathrm{6v})}$, hence $2D_{\mathrm{in}}$.

Channel-wise locked dropout randomly zeros entire channels uniformly across time steps, encouraging the model to distribute its representation broadly rather than overfitting to a subset of electrodes.
The core encoder is a stack of $5$ bidirectional GRU layers with layer normalization on the output:
\[
  \vb Y = (\operatorname{LN}\circ\operatorname{BiGRU}\circ\operatorname{Drop}_{\mathrm{chan}})(\vb X') \in \mathbb{R}^{B\times T\times 2D_{\mathrm{GRU}}},
\]
where $2D_{\mathrm{GRU}}$ reflects the concatenation of forward and backward hidden states.
The per-frame outputs are projected onto $\lvert\mathcal{P}\rvert=44$ phoneme tokens for connectionist temporal classification (CTC) decoding:
$\vb Z_{\mathrm{phon}} = \phi_{\mathrm{phon}}(\vb Y) \in \mathbb{R}^{B\times T\times \lvert\mathcal{P}\rvert}$.

\subsubsection{Parallel Brain Encoders for Multimodal Discovery}
In Stage~2, we add parallel brain encoders, each with a separate affine head mapping its outputs into the LLM embedding space:
$\vb Z_{\mathrm{LLM}}^{(m)} = \phi_{\mathrm{LLM}}^{(m)}(\vb Y^{(m)}) \in \mathbb{R}^{B\times T\times D_{\mathrm{LLM}}}$,
where superscript $(m)$ indexes the brain encoder modality.
In our full model, we have three brain encoders: one for area~6v, one for both areas, and one for area~44.
\begin{align*}
  &\vb Z_{\mathrm{LLM}}^{(m)} = \operatorname{BE}^{(m)}(\vb X^{(m)}),\quad m \in \{0, 1, 2\},\\
  \mathrm{where} \quad &\mathrm{BE}^{(m)} = \phi_{\mathrm{LLM}}^{(m)} \circ \operatorname{LN} \circ \operatorname{BiGRU} \circ \operatorname{Drop}_{\mathrm{chan}} \circ \operatorname{\phi}_{\mathrm{in}}^{(m,d)},\\
  \mathrm{and} \quad &\vb X^{(0)} =
  \left[\begin{array}{@{}c@{}}
    \vb c^{(\mathrm{6v})}_{1:T} \\\midrule \vb p^{(\mathrm{6v})}_{1:T}
  \end{array}\right],
  \quad \vb X^{(1)} =
  \left[\begin{array}{@{}c@{}} \vb p^{(\mathrm{6v})}_{1:T} \\\midrule \vb p^{(\mathrm{44})}_{1:T}
  \end{array}\right],
  \quad \vb X^{(2)} =
  \left[\begin{array}{@{}c@{}} \vb c^{(\mathrm{44})}_{1:T} \\\midrule \vb p^{(\mathrm{44})}_{1:T}
  \end{array}\right].
\end{align*}
$\operatorname{BE}^{(0)}$ is initialized from the pretrained Stage~1 encoder; the others share its architecture but are randomly initialized.
Due to limited computational resources, only three parallel encoders are used, and $\operatorname{BE}^{(1)}$ receives only spike power features, which are more robust and informative than threshold crossings for area~6v~\cite{willettHighperformanceSpeechNeuroprosthesis2023}.

\paragraph{Temporal Residual CNN}
The articulatory encoder $\operatorname{BE}^{(0)}$ captures phoneme-level motor dynamics, so its full per-frame output $\vb Z_{\mathrm{LLM}}^{(0)} \in \mathbb{R}^{B\times T\times D_{\mathrm{LLM}}}$ is retained as a token sequence for the LLM decoder.
To refine these representations, we apply a temporal residual convolutional neural network (CNN):
\[
  {\vb Z}_{\mathrm{LLM}}^{(0)}
  \overset{\mathrm{redef}}{=} (\phi_{\mathrm{ref}}\circ\operatorname{GELU}\circ\operatorname{ResCNN})(\vb Z_{\mathrm{LLM}}^{(0)}),
\]
where $\operatorname{ResCNN}(\vb Z) = \vb Z + (\operatorname{LN}\circ\operatorname{Conv}_2\circ\operatorname{GELU}\circ\operatorname{Conv}_1)(\vb Z)$ uses kernel size $5$ and dilation $2$, giving an effective receptive field of $17$ frames ($\approx1.56\,\mathrm{s}$ of input signal) that captures longer-range temporal context.

\paragraph{Unlabeled Utterance-Level Modalities}
\label{sec:utterance_level_modalities}
Unlike the articulatory encoder, $\operatorname{BE}^{(1)}$ and $\operatorname{BE}^{(2)}$ are designed to discover complementary modalities through self-supervised decorrelation, without phoneme supervision.
The information these encoders may capture---e.g., sentence-level semantics or syntactic structure---is not expected to be frame-aligned.
Accordingly, we represent each new modality as a single vector:
\[
  {\vb Z}_{\mathrm{LLM}}^{(m)}
  \overset{\mathrm{redef}}{=} \phi_{\mathrm{LLM}}^{(m)}(\vb Y^{(m)}_{T})
  \in \mathbb{R}^{B \times 1\times D_{\mathrm{LLM}}},
  \quad m \in \{1, 2\},
\]
where $\vb Y^{(m)}_{T}$ is the final hidden state of the BiGRU stack.
These vectors are appended to the articulatory token sequence as global conditioning signals for the LLM decoder.
Empirically, passing full sequence outputs from the new encoders destabilizes training---likely because multiple unsupervised token streams create redundant temporal signals with no phoneme-level anchor---whereas utterance-level summaries avoid this while preserving the capacity to encode abstract linguistic information.

\subsubsection{Multimodal LLM Decoder}

The decoder is a pretrained LLM (Aero-1-Audio-1.5B~\cite{Aero1AudioLMMsLab}) that autoregressively generates the transcription conditioned on the brain encoder outputs.
The input embedding sequence concatenates the encoder representations with a tokenized text prompt and, during training, the target transcription:
\[
  \vb E =
  \left[\begin{array}{c|c|c|c|c}
    \vb Z_{\mathrm{LLM}}^{(0)} & \vb Z_{\mathrm{LLM}}^{(1)} & \vb Z_{\mathrm{LLM}}^{(2)} & \vb E_{\mathrm{prompt}} & \vb E_{\mathrm{target}}
  \end{array}\right]
\]
where $\vb Z_{\mathrm{LLM}}^{(0)} \in \mathbb{R}^{B\times T\times D_{\mathrm{LLM}}}$ is the articulatory token sequence, $\vb Z_{\mathrm{LLM}}^{(1)}, \vb Z_{\mathrm{LLM}}^{(2)} \in \mathbb{R}^{B\times 1\times D_{\mathrm{LLM}}}$ are the utterance-level features, and $\vb E_{\mathrm{prompt}}, \vb E_{\mathrm{target}}$ are the token embeddings of the text prompt and target sentence.
Each encoder output occupies a separate conversational turn using the LLM's chat template, so the model can attend to each modality as a distinct input segment.
The text prompt instructs the model to transcribe the preceding neural signals.
During inference, $\vb E_{\mathrm{target}}$ starts empty and the model generates tokens autoregressively.

To adapt the LLM while preserving its pretrained language capabilities, we freeze all base parameters and apply QLoRA~\cite{huLoRALowRankAdaptation2021,dettmersQLORAEfficientFinetuning2023} with 4-bit NormalFloat quantization.
Low-rank adapters ($r{=}16$, $\alpha{=}32$) are inserted into all linear projections (\texttt{[qkvo]\_proj}, \texttt{(gate|up|down)\_proj}).
The brain encoders, projection heads, and LoRA adapters are the only trainable components in Stage~2.

\subsection{Training Objectives}

\subsubsection{Stage~1: Phoneme Decoding with CTC Loss}
\label{sec:ctc_loss}

In Stage~1 we train the articulatory encoder $\operatorname{BE}^{(0)}$ with the CTC loss~\cite{gravesConnectionistTemporalClassification2006} to predict phoneme sequences from area~6v neural activity.
Let $\vb p = (p_1, \dots, p_L)$ denote the target phoneme sequence and let $\mathcal{B}^{-1}(\vb p)$ be the set of frame-level alignments that reduce to $\vb p$ under the CTC collapsing function $\mathcal{B}$.
The loss marginalizes over all valid alignments:
\[
  \predisplaysize=0pt
  \mathcal{L}_{\mathrm{CTC}} = -\log \sum_{\mkern-9mu{\vb\pi \in \mathcal{B}^{-1}(\vb p)}\mkern-9mu} P(\vb\pi \mid \vb Z_{\mathrm{phon}}).
\]
Only $\operatorname{BE}^{(0)}$ and $\phi_{\mathrm{phon}}$ are trained in this stage; the LLM is not involved.

\subsubsection{Stage~2: End-to-End Training with Multimodal Discovery}
\label{sec:multimodal_losses}

In Stage~2, the full pipeline is trained end-to-end with a composite objective:
\[
  \predisplaysize=0pt
  \mathcal{L} = \mathcal{L}_{\mathrm{CE}}
  + \lambda_{\mathrm{con}}(s)\,\mathcal{L}_{\mathrm{con}}
  + \lambda_{\mathrm{dec}}(s)\,\mathcal{L}_{\mathrm{dec}},
\]
where $\mathcal{L}_{\mathrm{CE}}$ is the cross-entropy loss for transcription, $\mathcal{L}_{\mathrm{con}}$ is a contrastive alignment loss, $\mathcal{L}_{\mathrm{dec}}$ is a decorrelation loss, and $\lambda_{\mathrm{con}}(s), \lambda_{\mathrm{dec}}(s)$ are step-varying weights defined below.

\paragraph{Cross-Entropy Loss}
The LLM decoder is trained with standard next-token cross-entropy on the target transcription $\vb y = (y_1, \dots, y_L)$:
\[
  \predisplaysize=0pt
  \mathcal{L}_{\mathrm{CE}} = -\frac{1}{L}\sum_{l=1}^{L} \log P(y_l \mid \vb E_{<l}),
\]
where $\vb E_{<l}$ denotes all preceding input embeddings (neural tokens, prompt tokens, and previously generated target tokens).
This is the primary training signal that teaches the LLM to transduce neural representations into text.

\paragraph{Shared MoDAl Projection Space}
\label{sec:modal_space}
The contrastive and decorrelation losses operate in a shared projection space $\mathbb{R}^{D_\mathrm{MoDAl}}$ (the \emph{MoDAl space}), separate from the LLM embedding space to avoid disrupting its geometry.
For each brain encoder modality $m$ and for text $t$, we first obtain a single vector per sample by mean-pooling across time (a no-op for utterance-level $\operatorname{BE}^{(1,2)}$):
\[
  \bar{\vb z}^{(\cdot)}_i = \operatorname{pool}(\vb Z_{\mathrm{LLM},i}^{(\cdot)}) \in \mathbb{R}^{D_{\mathrm{LLM}}},
\]
then map each into the MoDAl space via modality-specific affine projections:
\[
  \vb u^{(\cdot)}_i = \phi^{(\cdot)}_{\mathrm{MoDAl}}(\bar{\vb z}^{(\cdot)}_i) \in \mathbb{R}^{D_{\mathrm{MoDAl}}},
  \quad (\cdot) \in \{0,1,2,t\}.
\]
% Projecting into a separate space avoids disrupting the LLM embedding geometry, which is optimized for next-token prediction rather than sample discrimination or modality decorrelation.

A critical design choice is that both losses share the same projection weights $\phi^{(\cdot)}_{\mathrm{MoDAl}}$ and per-sample L2 normalization, yielding \[\hat{\vb u}^{(\cdot)}_i = \vb u^{(\cdot)}_i / \lVert \vb u^{(\cdot)}_i \rVert \in \mathbb{S}^{D_{\mathrm{MoDAl}}-1}.\]
The contrastive loss operates on these unit vectors directly.
The decorrelation loss applies an additional per-feature batch standardization (centering and scaling each dimension to zero mean and unit variance across the batch) to obtain $\check{\vb u}^{(m)}_i$, so that the cross-correlation matrix $\vb C^{(m,n)}$ yields Pearson correlations.
If separate projectors were used for each loss, the decorrelation projector could map all samples to a constant vector---producing the zero matrix after batch standardization and trivially minimizing $\mathcal{L}_{\mathrm{dec}}$ without decorrelating the encoders.
Sharing projections prevents this: The contrastive loss requires sample-discriminative projections, precluding constant-map collapse and ensuring that decorrelation is computed over text-relevant directions (validated by ablations in Tables~\ref{tab:loss_ablation} and~\ref{tab:crosscorr}).
The shared L2 normalization further ensures both losses measure angular relationships in the same geometry, preventing an encoder from satisfying decorrelation through magnitude differences alone.

\paragraph{Contrastive Loss}
The contrastive loss aligns each brain encoder with the text representation via a symmetric InfoNCE objective~\cite{oordRepresentationLearningContrastive2019}.
For modality $m$ and sample $i$:
\[
  \ell^{(m \to t)}_{i} = -\log\frac{e^{\hat{\vb u}^{(m)}_i \cdot \hat{\vb u}^{(t)}_i / \tau}}
  {\sum_j e^{\hat{\vb u}^{(m)}_i \cdot \hat{\vb u}^{(t)}_j / \tau}},\quad
  \ell^{(t \to m)}_{i} = -\log\frac{e^{\hat{\vb u}^{(t)}_i \cdot \hat{\vb u}^{(m)}_i / \tau}}
  {\sum_j e^{\hat{\vb u}^{(t)}_i \cdot \hat{\vb u}^{(m)}_j / \tau}},
\]
where $\tau$ is a learnable temperature parameter.
The total contrastive loss symmetrizes over directions and averages over modalities:
\[
  \mathcal{L}_{\mathrm{con}} = \frac{1}{M}\sum_{m=0}^{M{-}1}
  \frac{1}{2B}\sum_{i=1}^{B}\bigl[\ell^{(m \to t)}_{i} + \ell^{(t \to m)}_{i}\bigr].
\]
Each modality is independently aligned with text; no direct inter-modality pairing is enforced.

\paragraph{Contrastive Transitivity}
Although the contrastive loss pairs each modality independently with text, all modalities share the same text projector $\phi^{(t)}_{\mathrm{MoDAl}}$ and therefore the same text targets $\hat{\vb u}^{(t)}_i$.
This creates a transitive alignment pressure: driving each $\hat{\vb u}^{(m)}_i$ toward $\hat{\vb u}^{(t)}_i$ simultaneously drives all modalities toward each other.

\begin{proposition}[Contrastive Transitivity]\label{prop:alignment}
  Optimize $\mathcal{L}_{\mathrm{con}}$ jointly over all modality projections $\{\hat{\vb u}^{(m)}_i\}$ and text projections $\{\hat{\vb u}^{(t)}_i\}$ on the unit sphere $\mathbb{S}^{D_{\mathrm{MoDAl}}-1}$ (the text targets are co-adapted, not held fixed).
  For $D_{\mathrm{MoDAl}} \ge B-1$ and any fixed temperature $\tau > 0$, every global minimizer satisfies
  $\hat{\vb u}^{(m)}_i = \hat{\vb u}^{(t)}_i$ for all $m \in \{0, \dots, M{-}1\}$ and $i \in [B]$.
  Consequently, $\hat{\vb u}^{(m)}_i = \hat{\vb u}^{(n)}_i$ for all modality pairs $(m,n)$ and all samples~$i$: All modalities coalesce.
\end{proposition}

\noindent
The proof (Appendix~\ref{app:proofs}) reduces $\mathcal{L}_{\mathrm{con}}$ to the neural-collapse objective of Lu and Steinerberger~\cite{luNeuralCollapseCrossentropy2022}; at the optimum the shared representation is in fact a simplex equiangular tight frame, unique up to a global orthogonal transformation, but only the coalescence $\hat{\vb u}^{(m)}_i = \hat{\vb u}^{(n)}_i$ is used below.

\paragraph{Decorrelation Loss}
To encourage the brain encoders to learn diverse representations, the decorrelation loss counteracts the coalescence described above by penalizing statistical dependence between modality pairs.
Starting from $\hat{\vb u}^{(m)}_i$, we apply per-feature batch standardization to obtain $\check{\vb u}^{(m)}_i$ and compute the empirical cross-correlation matrix:
\[
  \vb C^{(m,n)} = \frac{1}{B}\,\check{\vb U}^{(m)\top}\check{\vb U}^{(n)} \in \mathbb{R}^{D_{\mathrm{MoDAl}} \times D_{\mathrm{MoDAl}}},
\]
where $\check{\vb U}^{(m)} \in \mathbb{R}^{B \times D_{\mathrm{MoDAl}}}$ stacks the standardized projections and the diagonal entry $C^{(m,n)}_{ii}$ gives the Pearson correlation between feature $i$ across the two modalities.
The decorrelation loss penalizes only these diagonal entries:
\[
  \mathcal{L}_{\mathrm{dec}} = \frac{2}{M(M{-}1)} \sum_{0 \leq m < n < M}
  \frac{1}{D_{\mathrm{MoDAl}}} \sum_{i=1}^{D_{\mathrm{MoDAl}}} \bigl(C^{(m,n)}_{ii}\bigr)^2.
\]
This is inspired by Barlow Twins~\cite{zbontarBarlowTwinsSelfSupervised2021}, but adapted for inter-modality comparison: The cross-correlation is computed \emph{between} different modality encoders rather than between augmented views, and the diagonal target is zero (decorrelation) rather than one (invariance).
We omit the off-diagonal penalty $\sum_{i \neq j}(C^{(m,n)}_{ij})^2$: At the contrastive equilibrium where $\check{\vb U}^{(m)} = \check{\vb U}^{(n)}$, the off-diagonal entries reduce to \emph{within}-modality feature correlations, which regularize for feature independence but do not enforce the inter-modality decorrelation we seek.
The diagonal alone directly measures whether corresponding text-aligned features carry redundant information across modalities.

\begin{proposition}[Decorrelation at the Contrastive Fixed Point]\label{prop:decorr}
  If $\hat{\vb u}^{(m)}_i = \hat{\vb u}^{(n)}_i$ for all $i \in [B]$ (the contrastive equilibrium from Proposition~\ref{prop:alignment}), then after batch-wise feature standardization, $C^{(m,n)}_{ii} = 1$ for all $i \in [D_{\mathrm{MoDAl}}]$, and the decorrelation loss achieves its maximum value of $1$.
\end{proposition}

Propositions~\ref{prop:alignment} and~\ref{prop:decorr} (proofs in Appendix~\ref{app:proofs}) together establish that the contrastive and decorrelation losses are in \emph{productive tension}: The contrastive loss drives $C^{(m,n)}_{ii} \to 1$ (via transitive alignment through text), while the decorrelation loss drives $C^{(m,n)}_{ii} \to 0$.
At equilibrium, each encoder must maintain sufficient alignment with text to support the transcription task ($\mathcal{L}_{\mathrm{CE}}$) while encoding \emph{different} text-relevant information from the other encoders.
This is the mechanism by which modality discovery emerges without explicit modality labels.

\paragraph{Loss Warmup}
The contrastive and decorrelation loss weights are linearly warmed up from zero: $\lambda_{\mathrm{con}}(s) = (s/s_{\max})\,\lambda_{\mathrm{con}}$ and $\lambda_{\mathrm{dec}}(s) = (s/s_{\max})\,\lambda_{\mathrm{dec}}$, where $s$ is the current optimizer step and $s_{\max}$ the total number of steps.
This ensures that $\mathcal{L}_{\mathrm{CE}}$ establishes a baseline representation before the auxiliary losses begin to shape encoder specialization.

\section{Experiments}

\subsection{Experimental Setup}
\label{sec:setup}

\paragraph{Data Splits and Evaluation}
We use the default partitioning of the Brain-to-Text Benchmark~'24 dataset~\cite{willettHighperformanceSpeechNeuroprosthesis2023}: For each recording day, the first two blocks are reserved as competition holdout, the last block as the test set, and the remaining blocks as training data, yielding 8,780 training, 880 test, and 1,200 holdout sentences.
All models are selected based on test-set WER, and only the main experiment is evaluated on the competition holdout due to submission limits.
WER measures the minimum edit distance between predicted and reference word sequences, normalized by the reference length:
\[
  \mathrm{WER} = \frac{S + D + I}{N},
\]
where $S$, $D$, and $I$ denote the number of substitutions, deletions, and insertions, and $N$ is the number of reference words.
We lowercase all text and remove all punctuation except apostrophes before computing WER.

\paragraph{Training}
In Stage~1, the articulatory encoder $\operatorname{BE}^{(0)}$ is pretrained with the CTC loss for 300 epochs (peak learning rate $1{\times}10^{-4}$).
In Stage~2, the full pipeline is trained end-to-end for 150 epochs with the composite objective from Section~\ref{sec:multimodal_losses}.
We use two learning rate groups: $1.5{\times}10^{-4}$ for the pretrained $\operatorname{BE}^{(0)}$ and LoRA adapters, and $2{\times}10^{-4}$ for the randomly initialized modules ($\operatorname{BE}^{(1,2)}$ and projection heads), allowing the new encoders to adapt faster while using the pretrained modalities as stable anchors.
Both stages use AdamW~\cite{loshchilovDecoupledWeightDecay2019} with cosine annealing to zero.
The complete set of hyperparameters is given in Table~\ref{tab:hyperparams} (Appendix~\ref{app:hyperparams}).
All experiments are run on 8\,$\times$\,NVIDIA RTX PRO 6000 GPUs with distributed data-parallel training and mixed precision (bfloat16).

\subsection{Neural Speech Decoding Results}

\begin{table}
  \centering
  \caption{WER (\%) comparison on the Brain-to-Text Benchmark~'24. MoDAl results are mean $\bm\pm$\,95\% CI over 5 seeds. All methods use the same dataset and evaluation splits.}
  \label{tab:main_results}
  \begin{tabular}{llcc}
    \toprule
    \textbf{Method} & \textbf{Type} & \textbf{Test~$\downarrow$} & \textbf{Holdout~$\downarrow$} \\
    \midrule
    \citet{willettHighperformanceSpeechNeuroprosthesis2023}
    & Cascaded & 23.8--24.7\rlap{$^{\dagger}$} & \textbf{15.4} \\
    \citet{fengEndtoEndFrameworkInvasive2024}
    & E2E & 26.3 & 24.7 \\
    \cmidrule(lr){1-4}\addlinespace
    MoDAl-1 (ours) & E2E & \textit{22.4\,$\pm$\,0.2} & 18.8\,$\pm$\,0.5 \\
    \rowcolor{cmd!8}
    MoDAl-Full (ours) & E2E & \textbf{21.6\,$\pm$\,0.1} & \textit{17.7\,$\pm$\,0.4} \\
    \bottomrule
  \end{tabular}

  \smallskip\raggedright\footnotesize
  $^{\dagger}$\,Willett et al.\ report online WER on their own evaluation sentences (not the benchmark test split), separately for vocal (23.8\%; 95\% CI 21.8--25.9) and silent (24.7\%; 95\% CI 22.0--27.4) conditions. We include these as a reference point, noting that the evaluation sentences differ from the benchmark's test partition.
\end{table}

Table~\ref{tab:main_results} compares MoDAl against prior work.
We report \emph{MoDAl-1}, which uses only $\operatorname{BE}^{(0)}$ on area~6v with contrastive alignment to text, and \emph{MoDAl-Full}, which adds $\operatorname{BE}^{(1)}$ and $\operatorname{BE}^{(2)}$ with decorrelation and contrastive alignment.

The improvement over prior E2E work decomposes into two sources.
First, MoDAl-1 reduces WER by 3.9 percentage points on the test set and 5.9 on the holdout relative to \citet{fengEndtoEndFrameworkInvasive2024}, despite sharing the same two-stage paradigm (CTC pretraining followed by E2E LLM fine-tuning) and BiGRU encoder architecture.
This gain is attributable to contrastive alignment with text in the MoDAl space, which provides an auxiliary training signal that improves the encoder representations beyond what cross-entropy alone achieves.
Second, MoDAl-Full further reduces WER by 0.8 points on the test set and 1.1 on the holdout by incorporating the previously discarded area~44 signals (95\% confidence intervals do not overlap; Table~\ref{tab:main_results}).
This demonstrates that the MoDAl framework can extract complementary information from neural signals that prior methods found uninformative.

On the competition holdout, MoDAl (17.7\%) approaches the offline cascaded system of \citet{willettHighperformanceSpeechNeuroprosthesis2023} (15.4\%), despite fundamental architectural differences: Their system requires a separately trained $n$-gram language model with beam search re-scoring, whereas MoDAl is a single E2E model that relies on the LLM's pretrained weights for language modeling.
Their online system achieves 23.8--24.7\%; while not directly comparable due to different test partitions, MoDAl (21.6\%) outperforms this reference point.
Language model re-scoring is orthogonal to our modality discovery contribution and could in principle be applied on top of MoDAl to further improve performance.

\subsection{Ablation Studies}

\subsubsection{Loss Ablations}

\begin{table}
  \centering
  \caption{Loss ablation on the test set. WER (\%) is mean $\pm$\,95\% CI over 5 seeds. All configurations use three brain encoders.}
  \label{tab:loss_ablation}
  \Description{Table showing five loss configurations and their test WER. The full MoDAl model with all three losses in a shared projection space achieves the best WER of 21.6 percent.}
  \begin{tabular}{rccccc}
    \toprule
    & \textbf{$\mathcal{L}_{\mathrm{CE}}$} & \textbf{$\mathcal{L}_{\mathrm{con}}$} & \textbf{$\mathcal{L}_{\mathrm{dec}}$} & \textbf{Shared proj.} & \textbf{Test WER~$\downarrow$} \\
    \midrule
    1. &  \checkmark &            &            & n/a      & 23.9\,$\pm$\,0.7 \\
    2. &  \checkmark & \checkmark &            & n/a      & 22.0\,$\pm$\,0.1 \\
    3. &  \checkmark &            & \checkmark & n/a      & 26.1\,$\pm$\,0.6 \\
    4. &  \checkmark & \checkmark & \checkmark & $\times$ & 22.2\,$\pm$\,0.3 \\
    \rowcolor{cmd!8}
    5. & \checkmark  & \checkmark & \checkmark &\checkmark& \textbf{21.6\,$\pm$\,0.1} \\
    \bottomrule
  \end{tabular}
\end{table}

Table~\ref{tab:loss_ablation} isolates the contribution of each loss component.
Adding contrastive alignment to the cross-entropy baseline yields the largest single improvement (23.9\% $\to$ 22.0\%), confirming that aligning brain encoder representations with text in the MoDAl space substantially improves the learned representations beyond what $\mathcal{L}_{\mathrm{CE}}$ alone achieves.

Decorrelation without contrastive alignment (row~3) degrades performance to 26.1\%, worse than the CE-only baseline.
This validates the theoretical prediction: $\mathcal{L}_{\mathrm{dec}}$ is designed to counteract the transitive coalescence induced by $\mathcal{L}_{\mathrm{con}}$ (Propositions~\ref{prop:alignment}--\ref{prop:decorr}), and without that coalescence to push against, the decorrelation loss disrupts the shared structure that the LLM decoder relies on.

When both losses are active but the decorrelation projector is separate from the contrastive projector (row~4), performance is statistically indistinguishable from no decorrelation (row~2), as 95\% confidence intervals overlap.
Table~\ref{tab:crosscorr} confirms that the separate projector achieves near-identical cross-correlations to the no-decorrelation baseline---the trivial-collapse failure mode predicted in Section~\ref{sec:modal_space}.

Only when all three losses operate in the shared MoDAl space (row~5) does the full benefit emerge: 21.6\%, a further 0.4 percentage point improvement.
This confirms that productive tension between $\mathcal{L}_{\mathrm{con}}$ and $\mathcal{L}_{\mathrm{dec}}$ in a shared space is the mechanism that enables discovery of complementary modalities from the previously discarded area~44 signals.

We also find that the auxiliary losses are sensitive to batch size: Halving it to $B{=}64$ erases their benefit entirely, matching the CE-only baseline (Appendix~\ref{app:component_ablations}).

\subsubsection{Component Ablations}

\begin{table}
  \centering
  \caption{Component ablations on the test set. WER (\%) is mean $\pm$\,95\% CI over 5 seeds. All configurations use the full loss ($\mathcal{L}_{\mathrm{CE}} + \mathcal{L}_{\mathrm{con}} + \mathcal{L}_{\mathrm{dec}}$) in the shared MoDAl space.}
  \label{tab:component_ablation}
  \Description{Table showing four component ablations and the full model. Removing any encoder or using sequence output from the novel encoders degrades performance.}
  \begin{tabular}{lc}
    \toprule
    \textbf{Configuration} & \textbf{Test WER~$\downarrow$} \\
    \midrule
    1.\enspace Sequence output from $\operatorname{BE}^{(1,2)}$ & 24.0\,$\pm$\,0.3 \\
    2.\enspace $\operatorname{BE} = \{\operatorname{BE}^{(0)}, \operatorname{BE}^{(1)}\}$ & 22.2\,$\pm$\,0.3 \\
    3.\enspace $\operatorname{BE} = \{\operatorname{BE}^{(0)}, \operatorname{BE}^{(2)}\}$ & 22.3\,$\pm$\,0.3 \\
    4.\enspace $\operatorname{BE} = \{\operatorname{BE}^{(0)}, \operatorname{BE}^{(0)\prime}\}$ (area~6v only) & 22.5\,$\pm$\,0.3 \\
    \rowcolor{cmd!8}
    5.\enspace MoDAl-Full & \textbf{21.6\,$\pm$\,0.1} \\
    \bottomrule
  \end{tabular}
\end{table}

Table~\ref{tab:component_ablation} ablates the parallel encoder design.
Replacing the utterance-level representation of $\operatorname{BE}^{(1,2)}$ with full sequence output (row~1) causes the largest degradation (24.0\%), confirming that multiple unsupervised token streams destabilize training as discussed in Section~\ref{sec:utterance_level_modalities}: The novel encoders are expected to capture utterance-level linguistic structure rather than frame-aligned articulatory dynamics, and a single conditioning vector per encoder is sufficient to convey this information to the LLM decoder.

Dropping either novel encoder (rows~2--3) yields similar performance ($\approx$22.2--22.3\%), which is better than using no novel encoders at all (MoDAl-1: 22.4\% in Table~\ref{tab:main_results}) but worse than the full three-encoder model (21.6\%).
This indicates that $\operatorname{BE}^{(1)}$ and $\operatorname{BE}^{(2)}$ each contribute complementary information that the other does not fully recover---consistent with their different input compositions (area~6v$+$44 spike power vs.\ area~44 only).

Replacing $\operatorname{BE}^{(1,2)}$ with a second area~6v encoder $\operatorname{BE}^{(0)\prime}$ (row~4) performs worst among the two-encoder configurations (22.5\%), confirming that the benefit of additional encoders depends on access to area~44 signals, not merely on added model capacity.
Additional component ablations (ResCNN, dropout, batch size, LLM choice, learning rate) are reported in Appendix~\ref{app:component_ablations}.

\subsection{Analysis of Discovered Modalities}
\subsubsection{Decorrelation in the MoDAl Space}
\label{sec:decorr_analysis}

To verify that the decorrelation loss produces meaningfully different representations, we measure the mean squared diagonal cross-correlation $\frac{1}{D_{\mathrm{MoDAl}}}\sum_i (C^{(m,n)}_{ii})^2$ between each modality pair in the full $D_{\mathrm{MoDAl}}=8192$-dimensional MoDAl space on the test set.
This is the quantity minimized by $\mathcal{L}_{\mathrm{dec}}$; lower values indicate greater decorrelation.
We compare three configurations: (i)~no decorrelation loss ($\lambda_{\mathrm{dec}}=0$), (ii)~decorrelation through a separate projector (not shared with the contrastive loss), and (iii)~MoDAl-Full with decorrelation in the shared projection space.

\begin{table}
  \centering
  \caption{Mean squared diagonal cross-correlation $\frac{1}{D}\sum_i (C^{(m,n)}_{ii})^2$ between modality pairs on the test set ($D{=}D_{\mathrm{MoDAl}}{=}8192$, $B_\mathrm{test}{=}880$). Lower is better.}
  \label{tab:crosscorr}
  \Description{Table showing cross-correlation values for three modality pairs across three ablation configurations. MoDAl-Full achieves an order of magnitude lower correlation than the other two configurations.}
  \begin{tabular}{lccc}
    \toprule
    \textbf{Configuration~\textbackslash~Pair} & \textbf{${(0,1)}$~$\downarrow$} & \textbf{${(0,2)}$~$\downarrow$} & \textbf{${(1,2)}$~$\downarrow$} \\
    \midrule
    1.\enspace No $\mathcal{L}_{\mathrm{dec}}$  & .085 & .044 & .197 \\
    2.\enspace Separate proj.                    & .081 & .044 & .213 \\
    \rowcolor{cmd!8}
    3.\enspace MoDAl-Full                        & \textbf{.003} & \textbf{.004} & \textbf{.015} \\
    \bottomrule
  \end{tabular}
\end{table}

Table~\ref{tab:crosscorr} reveals two findings.
First, projecting into separate contrastive and decorrelation spaces (row~2) produces cross-correlations nearly identical to no decorrelation (row~1), confirming the trivial-collapse failure mode from Section~\ref{sec:modal_space}: The projector satisfies $\mathcal{L}_{\mathrm{dec}}$ in its own space without propagating decorrelation pressure to the encoders.
Second, the shared-space design (row~3) reduces cross-correlation by over an order of magnitude across all pairs, demonstrating that the contrastive loss prevents projector collapse and forces genuinely different encoder representations.
The corresponding WER differences (Table~\ref{tab:loss_ablation}) confirm that this decorrelation translates to improved downstream performance.

\subsubsection{Contrastive Alignment Quality}

We assess text--neural alignment quality in the MoDAl space using cosine-similarity retrieval and canonical correlation analysis (CCA).

\paragraph{Retrieval}
For each modality $m$, we use the L2-normalized MoDAl projections $\hat{\vb u}^{(m)}_i$ as queries to retrieve the matched text projection $\hat{\vb u}^{(t)}_i$ from the full test set via cosine similarity, and report recall@$k$.

\begin{figure}
  \centering
  \includegraphics[width=\linewidth]{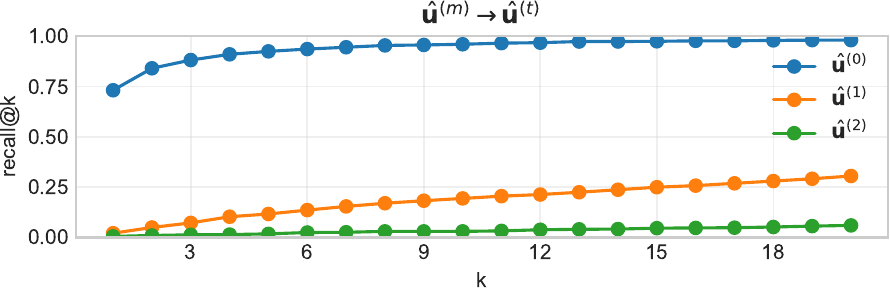}
  \caption{Neural-to-text retrieval recall@$k$ in the MoDAl space for each modality on the test set ($n=880$).}
  \Description[Recall at k curves for three modalities.]{Line plot showing recall at k from 1 to 20 for three modalities. Modality 0 achieves recall at 1 of 73\% rising to near 99\% at k equals 20. Modality 1 rises gradually from about 3\% to 30\%. Modality 2 remains below 7\% throughout.}
  \label{fig:retrieval}
\end{figure}

Figure~\ref{fig:retrieval} shows that $\operatorname{BE}^{(0)}$ achieves strong sample-level alignment (recall@1~$=73\%$, recall@20~$\approx 99\%$), expected given its rich temporal token sequence.
$\operatorname{BE}^{(1)}$ and $\operatorname{BE}^{(2)}$ show substantially lower retrieval (recall@1~$<5\%$), but this is consistent with the MoDAl design: These encoders produce a single utterance-level vector and are explicitly trained via $\mathcal{L}_{\mathrm{dec}}$ to avoid encoding the same sample-discriminative information as $\operatorname{BE}^{(0)}$.
Their role is as \emph{complementary conditioning signals} for the LLM---not standalone sample identifiers---confirmed by the WER improvement from MoDAl-1 to MoDAl-Full (Table~\ref{tab:main_results}) and the linguistic probe results below (Table~\ref{tab:probe_summary}).

\paragraph{Paired CCA}
To visualize the geometric relationship between neural and text representations beyond cosine retrieval, we apply CCA~\cite{hotellingRelationsTwoSets1936} to each neural--text pair.
Let $\hat{\vb z}^{(m)}_i = \bar{\vb z}^{(m)}_i / \lVert \bar{\vb z}^{(m)}_i \rVert$ denote the L2-normalized pooled representation in LLM space (before MoDAl projection).
CCA finds linear projections of the neural and text matrices into a shared 2-D subspace that maximizes pairwise correlation.
We apply it separately to the LLM-space pairs $(\hat{\vb z}^{(m)}, \hat{\vb z}^{(t)})$ and the MoDAl-space pairs $(\hat{\vb u}^{(m)}, \hat{\vb u}^{(t)})$.
Lines connect matched neural--text samples; shorter lines indicate tighter alignment.

\begin{figure}
  \centering
  \includegraphics[width=\linewidth]{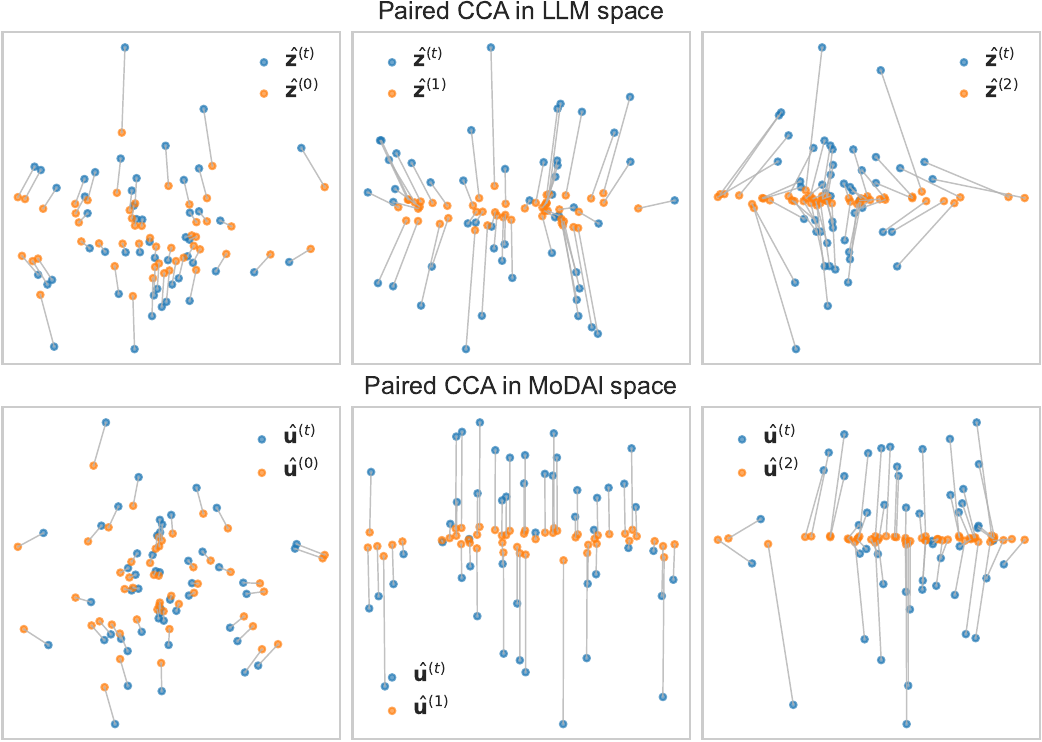}
  \caption{Paired CCA projections of neural and text representations. Top row: LLM embedding space ($\hat{\vb z}^{(m)}$ vs.\ $\hat{\vb z}^{(t)}$). Bottom row: MoDAl projection space ($\hat{\vb u}^{(m)}$ vs.\ $\hat{\vb u}^{(t)}$). Lines connect matched pairs ($n=50$ subsampled for clarity).}
  \Description[Six scatter plots showing paired CCA for three modalities in two spaces.]{Grid of six scatter plots arranged in two rows and three columns. Each plot shows blue dots for text and orange dots for neural representations connected by grey lines. In the top row showing LLM space, all three modalities show moderate alignment with scattered long lines. In the bottom row showing MoDAl space, modality 0 shows very tight pairing with short lines. Modalities 1 and 2 show clearer geometric structure with more organized pairing compared to the LLM space, though with longer lines than modality 0.}
  \label{fig:cca}
\end{figure}

Figure~\ref{fig:cca} compares alignment before (top, LLM space) and after (bottom, MoDAl space) projection.
In the LLM space, $\operatorname{BE}^{(0)}$ shows moderate correspondence and the novel modalities show noisy alignment.
After projection, the structure sharpens: $\operatorname{BE}^{(0)}$ exhibits tight pairing (short lines), while $\operatorname{BE}^{(1)}$ and $\operatorname{BE}^{(2)}$ reveal organized geometric patterns with systematic neural--text correspondence that is not visible in the LLM space.
This confirms that the MoDAl projection concentrates representational variance onto text-relevant directions, successfully extracting structure from the novel encoders that exists in the high-dimensional LLM embeddings but is not apparent without the learned projection.

\subsubsection{Representational Diversity}
To investigate whether the discovered modalities encode different linguistic information, we train linear probes~\cite{alainUnderstandingIntermediateLayers2018} on the L2-normalized MoDAl projections $\hat{\vb u}^{(m)}_i$ for each brain encoder and the text projection $\hat{\vb u}^{(t)}_i$ as a reference.
For each sentence in the test set ($n=880$), we extract eight linguistic properties from the ground-truth transcription using a rule-based grammar detector.
For classification targets (person, transitivity, tense, aspect, voice, determiner, wh-word) we fit a RidgeClassifier with stratified 5-fold cross-validation and report accuracy; for the regression target (sentence length) we fit Ridge regression ($\alpha{=}1.0$) and report $R^2$.
All results are significant at $p < 10^{-3}$ (one-sample $t$-test against chance).
Full details including class labels and per-fold standard deviations are given in Table~\ref{tab:probe} (Appendix~\ref{app:probes}).

\begin{table}
  \centering
  \caption{Linear probe scores on linguistic properties. Bold: best neural modality per row. Italics: second best. The text column ($\hat{\vb u}^{(t)}$) serves as a reference. Classification targets report accuracy; length reports $R^2$.}
  \label{tab:probe_summary}
  \begin{tabular}{llcccc}
    \toprule
    \textbf{Property} & \textbf{Metric~$\uparrow$}
    & \textbf{$\hat{\vb u}^{(0)}$} & \textbf{$\hat{\vb u}^{(1)}$} & \textbf{$\hat{\vb u}^{(2)}$}
    & \textbf{\textcolor{black!67}{$\hat{\vb u}^{(t)}$}} \\
    \midrule
    Length       & $R^2$  & .67          & \textbf{.87} & \textit{.75} & \textcolor{black!67}{.75} \\
    Voice        & Acc    & .90          & \textit{.95} & \textbf{.95} & \textcolor{black!67}{.91} \\
    Wh-word      & Acc    & .85          & \textit{.90} & \textbf{.90} & \textcolor{black!67}{.98} \\
    Determiner   & Acc    & .79          & \textbf{.83} & \textit{.83} & \textcolor{black!67}{.85} \\
    Person       & Acc    & \textbf{.74} & \textit{.71} & .59          & \textcolor{black!67}{.85} \\
    Aspect       & Acc    & .65          & \textbf{.68} & \textit{.67} & \textcolor{black!67}{.69} \\
    Tense        & Acc    & \textbf{.54} & .49          & \textit{.50} & \textcolor{black!67}{.73} \\
    Transitivity & Acc    & \textbf{.42} & \textit{.40} & .34          & \textcolor{black!67}{.48} \\
    \bottomrule
  \end{tabular}
\end{table}

Table~\ref{tab:probe_summary} reveals clear linguistic specialization.
The novel encoders $\operatorname{BE}^{(1)}$ (area~6v$+$44 spike power) and $\operatorname{BE}^{(2)}$ (area~44) achieve the highest neural scores on structural and syntactic properties: sentence length ($R^2 = .87$ for $\operatorname{BE}^{(1)}$, far exceeding even the text reference), grammatical voice (95\% for both), wh-word presence (90\%), and determiner type (83\%).
By contrast, $\operatorname{BE}^{(0)}$ (area~6v) excels at grammatical person (74\% vs.\ 71\% and 59\%) and tense (54\% vs.\ 49\% and 50\%); these are grammatical properties, but can be reliably inferred from articulatory patterns, since English pronouns (e.g., \emph{I}, \emph{you}) and tense morphology (e.g., \emph{-d}) are distinguishable from their articulatory signatures alone.

This pattern aligns with the neurolinguistic understanding of these cortical regions: Area~6v encodes motor commands for speech production, while area~44 (Broca's area) is associated with higher-level syntactic and structural processing.
The MoDAl framework successfully recovers this functional distinction without any explicit supervision over which linguistic properties each encoder should capture.
Notably, $\operatorname{BE}^{(1)}$ achieves $R^2 = .87$ on sentence length---substantially higher than the text reference (.75)---suggesting that this encoder captures temporal extent information from the neural signal that is not linearly accessible from pooled text embeddings.

\section{Conclusion}

We presented MoDAl, a framework for self-supervised neural modality discovery that combines contrastive alignment and decorrelation in a shared projection space.
The core insight is that these two objectives are in productive tension: Contrastive alignment anchors all brain encoders to text but transitively drives them to coalesce, while decorrelation counteracts this coalescence, forcing each encoder to specialize in different text-relevant information.
We proved that this tension is inherent to the shared-space design (Propositions~\ref{prop:alignment}--\ref{prop:decorr}) and showed empirically that both components---and their co-location in a single projection space---are necessary for the framework to discover useful modalities.

Applied to the Brain-to-Text Benchmark~'24, MoDAl achieved 21.6\% WER, a 4.7 percentage point improvement over the previous best E2E method.
The central finding is that area~44 (Broca's area) signals---discarded as uninformative by all prior work on this dataset---carry complementary linguistic information that the MoDAl framework can discover and exploit.
The contribution is both quantitative and qualitative: Beyond the WER reduction, linear probes reveal that the novel encoders surface syntactic and structural properties from area~44 that are complementary to the articulatory features captured by the primary encoder and not linearly accessible from text embeddings alone, consistent with the known neurolinguistic roles of these cortical regions.

Several directions remain open.
The current system uses three parallel encoders; scaling to more could reveal additional modalities.
The underlying mechanism---contrastive alignment with decorrelation for modality discovery---is domain-agnostic and could apply wherever intercorrelated signals must be decomposed into complementary representations, such as multi-sensor physiological monitoring or multi-region neural recordings for other cognitive tasks.
Finally, complementary decoding strategies such as language model re-ranking could be applied on top of MoDAl to further close the gap with cascaded systems.

\section*{Safe and Responsible Innovation Statement}

This work aims to restore communication for individuals with speech-impairing conditions such as ALS, directly supporting the inclusivity of multimodal interaction.
The neural data used are from a publicly available, de-identified benchmark dataset collected with informed consent under institutional review~\cite{willettHighperformanceSpeechNeuroprosthesis2023}.
Because neural signals may carry information beyond what participants intend to communicate, future deployments of speech BCIs must ensure that users retain control over when decoding is active and what is shared.
The MoDAl framework's ability to discover latent neural modalities underscores this concern: Decoded representations may reflect cognitive states beyond intended speech, and safeguards against unauthorized or covert decoding are essential.
We encourage the community to develop consent frameworks and privacy protections that evolve alongside the increasing capability of neural decoding systems.

% \begin{acks}
% \end{acks}
\begin{acks}
  This research was funded by the Defense Advanced Research Projects Agency (DARPA),
  under contract W912CG23C0031.
\end{acks}

\bibliographystyle{ACM-Reference-Format}
\bibliography{references}

\appendix

\section{Proofs}
\label{app:proofs}

\begin{proof}[Proof of Proposition~\ref{prop:alignment}]
  We optimize jointly over the modality projections $\{\hat{\vb u}^{(m)}_i\}$ and the text projections $\{\hat{\vb u}^{(t)}_i\}$ on the sphere, in the unconstrained-feature regime standard to this analysis; both are outputs of trained projectors (Section~\ref{sec:modal_space}), so the text targets are not fixed.

  \paragraph{Reduction to a neural-collapse functional}
  For a fixed modality $m$, write the neural-to-text term as
  \[
    \ell^{(m \to t)}_{i}
    = \log\frac{\sum_{j} e^{\hat{\vb u}^{(m)}_i \cdot \hat{\vb u}^{(t)}_j / \tau}}{e^{\hat{\vb u}^{(m)}_i \cdot \hat{\vb u}^{(t)}_i / \tau}}.
  \]
  Then $\sum_{i} \ell^{(m \to t)}_{i}$ is exactly the asymmetric neural-collapse objective of Lu and Steinerberger~\cite{luNeuralCollapseCrossentropy2022} on the two point sets $\{\hat{\vb u}^{(m)}_i\}_{i=1}^B$ and $\{\hat{\vb u}^{(t)}_i\}_{i=1}^B$, with inverse temperature $\alpha = 1/\tau$.

  \paragraph{One direction}
  By~\cite[Thms.~1--2]{luNeuralCollapseCrossentropy2022}, for $D_{\mathrm{MoDAl}} \ge B-1$ and any fixed $\tau > 0$, the global minimum of $\sum_{i} \ell^{(m \to t)}_{i}$ over unit-norm configurations is attained if and only if $\hat{\vb u}^{(m)}_i = \hat{\vb u}^{(t)}_i$ for all $i$ and $\{\hat{\vb u}^{(t)}_i\}$ forms a simplex equiangular tight frame.

  \paragraph{Symmetrization}
  The text-to-neural sum $\sum_{i} \ell^{(t \to m)}_{i}$ is the same functional with the two point sets interchanged, so it shares the identical minimizer set. Because a sum of two functionals with a common minimizer is minimized precisely on that set, $\mathcal{L}^{(m)}_{\mathrm{con}}$ attains its global minimum---denote it $\mathcal{L}^\star$---exactly on the aligned configurations $\hat{\vb u}^{(m)}_i = \hat{\vb u}^{(t)}_i$.

  \paragraph{Coupling through the shared text target}
  Since $\mathcal{L}_{\mathrm{con}} = \tfrac{1}{M}\sum_{m} \mathcal{L}^{(m)}_{\mathrm{con}}$ and every per-modality term shares the same text set $\{\hat{\vb u}^{(t)}_i\}$, each term satisfies $\mathcal{L}^{(m)}_{\mathrm{con}} \ge \mathcal{L}^\star$, so $\mathcal{L}_{\mathrm{con}} \ge \mathcal{L}^\star$ with equality if and only if each term equals $\mathcal{L}^\star$.
  Choosing one equiangular tight frame $\{\vb w_i\}_{i=1}^B$ and setting $\hat{\vb u}^{(t)}_i = \hat{\vb u}^{(m)}_i = \vb w_i$ for every $m$ and $i$ attains the bound.
  Conversely, any global minimizer must minimize every per-modality term simultaneously, which by the previous step forces $\hat{\vb u}^{(m)}_i = \hat{\vb u}^{(t)}_i$ onto a common frame for all $m$.
  Hence $\hat{\vb u}^{(m)}_i = \hat{\vb u}^{(n)}_i$ for all modality pairs $(m,n)$ and samples $i$---the modalities coalesce---unique up to a global orthogonal transformation of $\mathbb{R}^{D_{\mathrm{MoDAl}}}$.
\end{proof}

\paragraph{Fixed versus trainable temperature}
Proposition~\ref{prop:alignment} fixes $\tau$.
MoDAl instead learns $\tau$, and minimizing jointly over it would drive $\tau \to 0^+$, where the minimizer is no longer uniquely selected.
Under a trainable inverse temperature the global-minimum geometry is governed by the more permissive thresholdable condition of Bangachev and Polyanskiy~\cite{bangachevGlobalMinimizersSigmoid2026}, which no longer forces exact alignment and is consistent with an empirical modality gap.
This is by design rather than a defect: The trainable temperature and $\mathcal{L}_{\mathrm{dec}}$ are the counterforces that keep the coalescence pressure of $\mathcal{L}_{\mathrm{con}}$ from collapsing the encoders.

\begin{proof}[Proof of Proposition~\ref{prop:decorr}]
  If $\hat{\vb u}^{(m)}_i = \hat{\vb u}^{(n)}_i$ for all $i$, then the columns
  $[\hat{\vb u}^{(m)}_1, \dots, \hat{\vb u}^{(m)}_B]^\top$ and
  $[\hat{\vb u}^{(n)}_1, \dots, \hat{\vb u}^{(n)}_B]^\top$ are identical matrices.
  Since batch standardization is a deterministic column-wise affine transformation,
  $\check{\vb U}^{(m)} = \check{\vb U}^{(n)}$.
  Therefore, for each feature $i$,
  \[
    C^{(m,n)}_{ii} = \frac{1}{B}\sum_{k=1}^B \check{U}^{(m)}_{ki}\check{U}^{(n)}_{ki}
    = \frac{1}{B}\sum_{k=1}^B (\check{U}^{(m)}_{ki})^2 = 1,
  \]
  where the last equality follows from the unit-variance property of batch standardization.
\end{proof}

\section{Hyperparameters}
\label{app:hyperparams}

Table~\ref{tab:hyperparams} lists all hyperparameters used in the final model.
Stage~1 and Stage~2 share the same brain encoder architecture; differences are noted where applicable.

\begin{table}
  \centering
  \caption{Hyperparameters for the final MoDAl model.}
  \label{tab:hyperparams}
  \begin{tabular}{l l}
    \toprule
    \textbf{Hyperparameter} & \textbf{Value} \\
    \midrule
    \multicolumn{2}{l}{\textit{Brain encoder (all stages)}} \\
    Input kernel size $k$ / stride $s$ & 14 / 4 \\
    Channel-wise locked dropout & 0.75 \\
    BiGRU layers & 5 \\
    GRU hidden dim (per direction) & 512 \\
    GRU dropout & 0.75 \\
    \midrule
    \multicolumn{2}{l}{\textit{ResCNN (Stage 2, $\operatorname{BE}^{(0)}$ only)}} \\
    Conv kernel size / dilation & 5 / 2 \\
    \midrule
    \multicolumn{2}{l}{\textit{LLM decoder (Stage 2)}} \\
    LLM embedding dim $D_{\mathrm{LLM}}$ & 1536 \\
    QLoRA quantization & NF4, double quant \\
    LoRA rank $r$ / scaling $\alpha$ & 16 / 32 \\
    LoRA dropout & 0.15 \\
    LoRA target modules & all \texttt{*\_proj} layers \\
    \midrule
    \multicolumn{2}{l}{\textit{MoDAl projection (Stage 2)}} \\
    MoDAl projection dim $D_{\mathrm{MoDAl}}$ & 8192 \\
    Projector type & Affine \\
    Temperature $\tau$ (initial) & 0.1 \\
    $\lambda_{\mathrm{con}}$ (final) & 1.0 \\
    $\lambda_{\mathrm{dec}}$ (final) & 2.0 \\
    \midrule
    \multicolumn{2}{l}{\textit{Optimization}} \\
    Optimizer & AdamW \\
    $(\beta_1, \beta_2)$ & (0.9, 0.999) \\
    Weight decay & $10^{-2}$ \\
    LR schedule & Cosine annealing \\
    Peak LR (Stage~1) & $1 \times 10^{-4}$ \\
    Peak LR (Stage~2, $\operatorname{BE}^{(0)}$ only) & $1.5 \times 10^{-4}$ \\
    Peak LR (Stage~2, others) & $2.0 \times 10^{-4}$ \\
    Batch size & 128 \\
    Stage~1 epochs & 300 \\
    Stage~2 epochs & 150 \\
    Gradient clipping & 10.0 \\
    \bottomrule
  \end{tabular}
\end{table}

% moved table up to avoid it dangling on a new page
\begin{table*}
  \centering
  \caption{Full linear probe results. Score is mean $\pm$ std over 5 folds. All $p < 10^{-3}$.}
  \label{tab:probe}
  \begin{tabular}{l l l c c c}
    \toprule
    \textbf{Property} & \textbf{Task} & \textbf{Classes / Range}
    & $\hat{\vb u}^{(0)}$ & $\hat{\vb u}^{(1)}$ & $\hat{\vb u}^{(2)}$
    \\
    \midrule
    Length       & Regression     & 2--17 words
    & $.67 \pm .03$ & $.87 \pm .02$ & $.75 \pm .03$ \\
    Voice        & Classification & active, passive, none
    & $.90 \pm .01$ & $.95 \pm .00$ & $.95 \pm .00$ \\
    Wh-word      & Classification & true, false
    & $.85 \pm .04$ & $.90 \pm .01$ & $.90 \pm .01$ \\
    Determiner   & Classification & definite, indefinite, other, none
    & $.79 \pm .02$ & $.83 \pm .01$ & $.83 \pm .01$ \\
    Person       & Classification & 1st, 2nd, 3rd, none
    & $.74 \pm .03$ & $.71 \pm .02$ & $.59 \pm .02$ \\
    Aspect       & Classification & simple, continuous, perfect, none
    & $.65 \pm .03$ & $.68 \pm .02$ & $.67 \pm .02$ \\
    Tense        & Classification & past, present, future, none
    & $.54 \pm .02$ & $.49 \pm .04$ & $.50 \pm .03$ \\
    Transitivity & Classification & intransitive, transitive, ditransitive, impersonal, none
    & $.42 \pm .01$ & $.40 \pm .03$ & $.34 \pm .03$ \\
    \bottomrule
  \end{tabular}
\end{table*}

\section{Additional Component Ablations}
\label{app:component_ablations}

Table~\ref{tab:additional_ablations} reports ablations of engineering components not directly related to the modality discovery mechanism.

\begin{table}
  \centering
  \caption{Additional component ablations on the test set. WER (\%) is mean $\pm$\,95\% CI over 5 seeds. MoDAl-Full (21.6\,$\pm$\,0.1) is the reference.}
  \label{tab:additional_ablations}
  \begin{tabular}{lc}
    \toprule
    \textbf{Configuration} & \textbf{Test WER~$\downarrow$} \\
    \midrule
    w/o ResCNN                          & 25.0\,$\pm$\,0.2 \\
    w/o LayerNorm on GRU                & 22.1\,$\pm$\,0.2 \\
    w/o $\operatorname{Drop}_{\mathrm{chan}}$ & 26.4\,$\pm$\,0.3 \\
    LLM=Qwen2-Audio-7B-Instruct               & 27.6\,$\pm$\,0.4 \\
    Batch size $B{=}64$                 & 23.9\,$\pm$\,0.3 \\
    Single LR ($1.5{\times}10^{-4}$)    & 22.2\,$\pm$\,0.2 \\
    Single LR ($2{\times}10^{-4}$)      & 22.2\,$\pm$\,0.2 \\
    \bottomrule
  \end{tabular}
\end{table}

The ResCNN on $\operatorname{BE}^{(0)}$ contributes substantially (25.0\% without vs.\ 21.6\% with), confirming that the temporal refinement of articulatory features is important for the LLM decoder.
Channel-wise locked dropout is critical (26.4\% without), consistent with its role in preventing electrode overfitting across recording sessions.
LayerNorm on the GRU output has a modest effect (22.1\% vs.\ 21.6\%).

Using Qwen2-Audio-7B-Instruct \cite{chuQwen2AudioTechnicalReport2024} as the LLM decoder degrades performance to 27.6\%.
The larger model requires substantially more GPU memory, preventing us from increasing the LoRA rank beyond $r{=}16$ or the batch size beyond $B{=}128$ under our hardware constraints.
A higher-capacity adapter and larger batches may be necessary to effectively adapt a 7B model with QLoRA on this dataset; we leave this exploration to future work.

Halving the batch size from 128 to 64 degrades WER to 23.9\%.
Both the contrastive and decorrelation losses benefit from larger batches: The InfoNCE objective requires sufficient in-batch negatives for effective sample discrimination, and the cross-correlation matrix in $\mathcal{L}_{\mathrm{dec}}$ provides more stable estimates of inter-modality feature correlation with more samples.

Using a single learning rate for all Stage~2 modules ($1.5{\times}10^{-4}$ or $2{\times}10^{-4}$) slightly degrades performance to 22.2\% in both cases, confirming that the decoupled schedule---slower for pretrained components, faster for randomly initialized ones---provides a modest but consistent benefit.

\section{Linear Probe Details}
\label{app:probes}

Table~\ref{tab:probe} reports the full linear probe results.
All probes operate on the L2-normalized MoDAl projections ($\hat{\vb u}^{(m)}$ for neural, $\hat{\vb u}^{(t)}$ for text) from the test set (880 sentences).
Linguistic properties are extracted from ground-truth transcriptions using the \texttt{grammar-detector} Python library \cite{GrammardetectorGrammaticalFeature}.
Sentences for which parsing fails (19 out of 880) receive default \texttt{none} labels (e.g., a short utterance like `summary points' does not have tense).
For classification, we use \texttt{RidgeClassifier} with \texttt{StandardScaler} and stratified 5-fold cross-validation, reporting accuracy.
For regression (sentence length), we use \texttt{Ridge} ($\alpha=1.0$) with \texttt{StandardScaler} and 5-fold cross-validation, reporting $R^2$.
$p$-values are computed via one-sample $t$-test of the 5 fold scores against the chance baseline (uniform accuracy $1/\lvert\mathrm{classes}\rvert$ for classification, $R^2 = 0$ for regression).

\end{document}